\documentclass[journal]{IEEEtran}
\usepackage{blindtext}
\usepackage{graphicx}
\usepackage{color}
\usepackage{caption}
\usepackage{subcaption}
\usepackage{amsmath}
\usepackage{amssymb}
 \usepackage{amsthm}
 \usepackage{float}
 \usepackage{tikz}
\usepackage{url}

\newtheorem{thm}{Theorem}[section]

\newtheorem{lem}[thm]{Lemma}
\newtheorem{prop}[thm]{Proposition}
\theoremstyle{definition}
\theoremstyle{remark}

\newtheorem{dftn}{Definition}
\newtheorem{ques}{Question}

\ifCLASSINFOpdf
\else
\fi

\usepackage{amsmath}

\usepackage{color}

\begin{document}
%
\title{Efficient Quartet Representations of Trees and Applications to Supertree and Summary Methods}
%
%
%

\author{Ruth~Davidson,
 MaLyn~Lawhorn,
 Joseph~Rusinko*,
 and~Noah~Weber
\thanks{R. Davidson is with the University of Illinois Urbana-Champaign. email: redavid2@illinois.edu }
\thanks{M. Lawhorn and N. Weber are with Winthrop University. e-mail: lawhornc2@winthrop.edu and webern2@winthrop.edu}
\thanks{J. Rusinko is with Hobart and William Smith Colleges. email: rusinko@hws.edu}
\thanks{*Corresponding Author}}

%
%

\markboth{}%
{}
%


\maketitle

\begin{abstract} Quartet trees displayed by larger phylogenetic trees have long been used as inputs for species tree and supertree reconstruction. Computational constraints prevent the use of all displayed quartets in many practical problems with large numbers of taxa. We introduce the notion of an Efficient Quartet System (EQS) to represent a phylogenetic tree with a subset of the quartets displayed by the tree. We show mathematically that the set of quartets obtained from a tree via an EQS contains all of the combinatorial information of the tree itself. Using performance tests on simulated datasets, we also demonstrate that using an EQS to reduce the number of quartets in both summary method pipelines for species tree inference as well as methods for supertree inference results in only small reductions in accuracy.
\end{abstract} 

\section{Introduction}\label{Introduction}
Phylogenetic reconstruction algorithms turn a single set of input data about a set of taxa into a tree that reflects the evolutionary relationships among the taxa. Due to advances in molecular sequencing technology in recent decades, the input to a phylogenetic reconstruction problem is usually a set of molecular sequences.

Our contribution in this manuscript is useful for \emph{two-step pipelines} that first infer an alignment of the sequence data \cite{muscle,ClustalOmega,PRANK} and then infer a tree from the alignment \cite{RAxML,FastTree2,NJ,FastME2}. Inputs to two-step phylogenomic inference methods are usually aligned molecular sequence data from genomes of taxa obtained as short sequences of nucleotides or proteins referred to as \emph{genes}. The gene alignments are either used to infer \emph{gene trees}, which are then used as the input to phylogenomic methods known as \emph{summary methods}, or concatenated into a single long alignment before applying a phylogenetic inference method to infer a species tree. The concatenation versus summary method approach is the subject of lively debate \cite{Delusion} that lies outside the scope of this paper. 

Combining a collection of gene trees into a single tree representing the relationships among the species is known as the gene-species tree problem \protect\cite{maddison1997gene}. The relationship between gene and species trees can be modeled by the multi-species coalescent (MSC) \cite{Kingman82, PamiloNei, Tajima83}. The MSC provides a theoretical basis for advances (see \cite{IdentifyingRooted, DegnanRosenberg06}) in the development of species tree reconstruction methods such as \cite{mirarab2015astral}.

Reconstructing the entire tree of life may necessitate combining a collection of species trees to construct a \emph{supertree} that reflects the relationships among a larger set of taxa. This process is called \emph{supertree reconstruction.}  Finding an unrooted supertree that is maximally consistent with a set of input trees is computationally difficult: even determining whether a set of unrooted trees is compatible is NP-complete \protect\cite{steel1992}. As a result traditional supertree reconstruction algorithms are currently limited in scale \protect\cite{zimmermann2014bbca,wickett2014phylotranscriptomic}.

One approach to handling the computational challenges in both species tree and supertree reconstruction is through the analysis of four-taxon subtrees known as quartets \cite{origqpuz, warnow, snir, WQMC}.  In this approach one identifies quartet relationships displayed by the individual gene trees or incomplete species trees, and then combines these quartets through a quartet-agglomeration algorithm into a single tree that reflects the observed relationships. Such strategies include heuristics used as summary methods \cite{WQMC} or exact algorithms allowing a constrained set of possible species tree outputs \cite{mirarab2015astral} that are designed to produce a species tree displaying the maximum number of quartets displayed by the set of input trees given to the method.  

In \cite{IdentifyingRooted} it was shown that under the probability distribution induced by the MSC model the most frequent unrooted quartet tree matched the unrooted shape of the species tree.  This gives a theoretical motivation to develop quartet-based methods of species tree inference. In supertree reconstruction it is assumed that the subtrees have consistent topologies with the supertree and thus should share the same set of quartet relationships \protect\cite{supertree}.

Thus quartet-agglomeration remains a popular technique in phylogenetic reconstruction despite the fact that the Maximum Quartet Consistency Problem is known to be an NP-hard optimization problem \cite{steel1992}. Effective heuristics exist for combining quartets such as Quartets MaxCut (QMC) \cite{snir} and the recent modification of QMC, wQMC \cite{WQMC}. QMC and wQMC are popular due to their speed and, as we will discuss in this manuscript, their accuracy under simulation tests. However, there are limits to the size of a set of taxa in an inference problem that can be handled by these methods.  For example, the work of Swenson et al. shows that QMC using all the quartets fails to return an answer using $500$ taxa, as does Matrix Representation with Parsimony (MRP) \protect\cite{ragan1992phylogenetic}.

One strength of quartet-based reconstruction is that in order to reconstruct an $n$-taxon tree one does not need all ${n \choose 4}$ input quartets. Theoretically, a carefully selected set of $n-3$ quartets is sufficient, but this requires knowledge of the correct tree \protect\cite{bocker1999patching}. In practice, some studies have indicated that randomly sampled quartets on the order of $n^3$ are sufficient for reliable reconstruction \protect\cite{snir}.
However, even QMC using quartets sampled via a stochastic method can fail once the number of taxa approaches $1000$  due to the overwhelming number of quartets that must be analyzed \protect\cite{supertree}. This random sampling approach is also used in the biological analysis in \cite{SVD} as well as the simulations in \cite{WQMC}. 

The number of quartets under consideration affects the running time of any phylogenetic inference method that takes quartets as inputs. Therefore it is natural to ask if a small subset of quartets can be used without losing any mathematical information about the tree.

In this paper we propose a method of quartet sampling that is based on the combinatorics of definitive quartets. A collection of trees is called \emph{definitive} if there exists a unique tree that displays all of the trees in the collection. Our method builds on the combinatorial structures developed in \protect\cite{bocker1999patching,linked}, and proposes sampling a particular set of input quartets that we call an \emph{Efficient Quartet System} (EQS).

An EQS is definitive and thus captures all of the phylogenetic signal contained in the input trees. Since QMC is a heuristic algorithm with no theoretical guarantees it does not always return the correct tree even when the input quartets are definitive (see \protect\cite{linked} for a six-taxon example). However, we demonstrate that QMC returns the correct simulated model tree given an EQS as input with extremely high probability. We test the efficacy of using an EQS representation of a set of gene trees as an input for summary methods as well as the efficacy of using an EQS representation of the incomplete species trees used in supertree reconstruction. 

Sampling using an EQS is an alternative to sampling \emph{short quartets}, or those with a smaller diameter in the input tree, with larger probability than sampling quartets at random \protect\cite{warnow}. The short quartets approach prioritizes the inclusion of quartets thought to be accurately reconstructed, while our approach prioritizes selecting quartets that are guaranteed to retain the combinatorial features of the inferred tree.

\section{Efficient Quartet Systems}
\label{sec:EQS}

\subsection{Theoretical Properties of an EQS}

A \emph{phylogenetic tree} is a connected, undirected, acyclic graph in which the vertices of degree one (often called \emph{leaves}) are labeled by a set of taxa. We assume all trees in this paper are \emph{binary}, meaning that the internal vertices are of degree three.  Furthermore, we assume all trees in this paper are unrooted.  A \emph{quartet tree} is a phylogenetic tree with four leaves. We write $ab|cd$ when $a$ and $b$ form a \emph{cherry} of $T$. Figure~\protect\ref{fig:quartet} shows this tree.

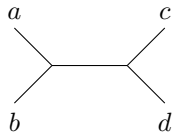
\begin{figure}[h]
\centering
\begin{tikzpicture}[scale=.5]
\draw (-1,0)--(1,0);
\draw (1,0)--(2,1);
\draw (1,0)--(2,-1);
\draw (-1,0)--(-2,1);
\draw (-1,0)--(-2,-1);
\node [above] at (2,1) {$c$};
\node [above] at (-2,1) {$a$};
\node [below] at (2,-1) {$d$};
\node [below] at (-2,-1) {$b$};
\end{tikzpicture}
\caption{Unrooted Quartet $ab|cd$}
\label{fig:quartet}
\end{figure}

A quartet is the fundamental unit of evolutionary information when working with methods based on time-reversible models of sequence evolution such as GTR \cite{GTR}. 

Let $T$ be a phylogenetic tree on a set of taxa $\mathcal{X}$. For every subset $\mathcal{S} \subseteq \mathcal{X}$, the \emph{induced subtree} $T|\mathcal{S}$ is the tree constructed by taking the unique minimal connected subgraph of $T$ containing the leaves in $\mathcal{S}$, and then removing all vertices of degree two. The \emph{support of a tree}, denoted by $supp(T)$, is the collection of taxa at the leaves of $T$. We say a tree $T_1$ \emph{displays} a tree $T_2$ if $T_1|supp(T_2)=T_2$. The set of quartets of a tree, denoted by $Q(T)$, is the collection of quartets displayed by $T$.

The following definition generalizes the notion of a quartet distinguishing an edge of a tree (cf. def. 6.8.3 in \protect\cite{semple}). 

\begin{dftn} A quartet $q=ab|cd$ \emph{distinguishes a path $p$} between internal vertices $v_1$ and $v_2$ of $T$ if the following three conditions are met:
\begin{enumerate}
\item $\{a,b\}$ and $\{c,d\}$ are subsets of different connected components of graph formed by removing the path $p$ from the tree $T$,
\item the path between $a$ and $b$ in $T$ passes through $v_1$, and 
\item the path between $c$ and $d$ in $T$ passes through $v_2$.
\end{enumerate} 
\end{dftn}
We define a representative subset of $Q(T)$ known as an \emph{Efficient Quartet System} (EQS) which is both definitive and contains a quartet which distinguishes a path between each pair of internal vertices of $T$. To construct an EQS we first assign to each internal vertex a \emph{representative set of taxa} (RST). To do so we first observe that each internal vertex on a binary tree partitions the taxa into three disjoint sets $S_1$, $S_2$, and $S_3$. 

\begin{dftn}
For a given ordering of the internal vertices of a tree, sequentially assign each vertex $v_i$ a three-element \emph{representative set of taxa} denoted by $RST(v_i)$ consisting of \emph{representative taxa} $rt_1(v_i) \in S_1(v_i)$, $rt_2(v_i) \in S_2(v_i)$, and $rt_3(v_i) \in S_3(v_i)$ which are the fewest number of edges from $v_i$. When there are multiple taxa satisfying these conditions we use the following tie-breaking procedure:
\begin{itemize}
\item Choose a taxon that is part of a cherry.
\item Select the taxon appearing in the most 
RSTs for the preceding vertices.
\item Select a taxon at random.
\end{itemize} 
\end{dftn}

The tie-breaking procedure in the second bullet point ensures Lemma \ref{lem:rtssmall} holds. 
\begin{lem}
\label{lem:rtssmall}
If $v_i$ and $v_j$ are adjacent internal vertices of a tree $T$, then $|RST(v_i) \cup RST(v_j)|=4$.
\end{lem}
\begin{proof}
Let $v_i$ and $v_j$ be adjacent internal vertices of a tree $T$.  Assume that $i$ precedes $j$  in the ordering of internal vertices used to compute the associated RST.  The removal of each vertex partitions the tree into three connected components. We assume that $S_1(v_i)$ (resp. $S_1(v_j)$) corresponds to the taxa in the connected component of the tree that contains $v_j$ (resp. $v_i$).  Since any path which from $v_i$ to a taxa in $S_1$ must pass first through $v_j$ it follows from the second bullet point in the tie-breaking procedure that $rt_1(v_i)$  must be either $rt_2(v_j)$ or $rt_3(v_j)$. Therefore  $rt_1(v_i) \in RST(v_j)$.  A similar argument shows  that $rt_1(v_j) \in RST(v_i)$.  The proof then follows from noting that $|RST(v_i) \cup RST(v_j)|>3$ since both $rt_2(v_j)$ and $rt_3(v_j)$ are elements of $S_1(v_i)$ and thus cannot both be in $RST(v_i)$.
\end{proof}

Given a choice of RSTs, we construct a collection of quartets known as \emph{efficient quartets}.

\begin{dftn}
Given a pair of internal vertices $v_i$ and $v_j$ and a RST, the associated \emph{efficient quartet} is the unique quartet $q=ab|cd$ such that $supp(q) \subset RST(v_i) \cup RST (v_j)$, and which distinguishes the path between $v_i$ and $v_j$.  Explicitly, $a=rt_2(v_i),b=rt_3(v_i),c=rt_2(v_j)$, and $d=rt_3(v_j)$.
\end{dftn}

\begin{dftn}
Given a fixed RST of $T$, an \emph{Efficient Quartet System} (EQS) of $T$-denoted $E(T)$-is the set of all possible efficient quartets associated to $T$.
\end{dftn}
Since the RST is dependent on both an ordering of internal vertices and the potential random selection of taxa there is not a unique EQS associated to $T$ thus we refer to \emph{an} EQS rather than \emph{the} EQS.  However, it is important to note that the cardinality of an EQS $E(T)$ for any tree $T$ is independent of these choices.

\begin{lem}\label{lem:EQS}
In a tree $T$ with $n$ taxa and any EQS representation $E(T)$ of $T$, $|E(T)| = {n-2 \choose 2}$.
\end{lem}
\begin{proof}
This is clear because there is one efficient quartet for each pair of internal vertices in the tree. \end{proof}
So, as $|Q(T)|={n \choose 4}$ and $|E(T)| = {n-2 \choose 2}$, reconstruction pipelines incorporating an EQS will require the use of $O(n^{2})$ fewer quartets.

Small definitive systems of quartets are strong candidates for supertree inputs as they retain all of the combinatorial information from the input trees. One example of a family of small definitive systems of quartets is known as \emph{linked systems}, introduced in \protect\cite{linked}.

\begin{dftn}[\cite{linked}]
Given a subset of quartets $L \subset Q(T)$ of $n-3$ quartet trees, define the associated graph $G_T(L)$ with vertex set $V$ and edge set $E$ as follows:
\begin{itemize}
\item the vertex set $V$ is the set of all quartet trees $q \in L$ which distinguish a unique edge in $T$, and
\item vertex pairs $\{q_i,q_j\}$ are connected by an edge $e \in L$ if the edge $e_i$ that $q_i$ distinguishes is adjacent to the edge $e_j$ that $q_j$ distinguishes and $|supp(\{q_i,q_j\})|=5$.
\end{itemize}
Two quartets are \emph{linked} if their vertices are connected in $G_T(L)$. The system of quartet trees $L$ is \emph{a linked system} if $G_T(L)$ is connected. 
\end{dftn}

\begin{prop}\label{prop:EQS}
An EQS is a \textit{definitive} set of quartets.
\end{prop}

\begin{proof}
Let $E(T)$ be an EQS for a tree $T$. We denote by $L(T)$ to be the subset of $E(T)$ of size $n-3$ of quartets that distinguish edges of $T$. Let $q_i$ and $q_j$ be elements of $L(T)$ which distinguish adjacent edges. It follows from Lemma~\ref{lem:rtssmall} that $|supp(q_i) \cup supp(q_j)|\le5$. This must in fact be an equality since $q_i \ne q_j$. Thus $L(T)$ is a linked system of quartets. It follows from Theorem 3.1 in \protect\cite{linked} that $L(T)$ is a definitive set of quartets. Since each quartet in $E(T)$ is displayed by a tree, then $T$ must be the unique tree that displays the quartets in $L(T)$. Therefore $E(T)$ is a definitive set of quartets.
\end{proof}

Along with Lemma \ref{lem:EQS}, Proposition \ref{prop:EQS} indicates that the notion of an EQS  satisfies theoretical properties that merit exploration in phylogenetic inference pipeline applications. 

\section{Experimental Properties of an EQS}
\label{sec:analysis}
In our first experiment we choose an input tree $T$ and ask if QMC returns a tree similar to $T$ given the input $E(T)$. This baseline test indicates that not too much information is being lost in the conversion of input data from $T$ to $E(T)$. In the second and third experiments we test the utility of adding an EQS component to pipelines for species tree construction and supertree construction.  We use wQMC in both of these experiments.

\subsection{Baseline Finding}
\label{sec:baseline}
Given a tree $T$, we denote by $QMC(E(T))$ the tree constructed by applying $QMC$ to the efficient representation $E(T)$ of $T$.
To measure the amount of information lost by using $E(T)$ to represent $T$ and reconstruct $T$ using QMC, we generated a tree $T$ and then computed the Normalized Robinson-Foulds (RF) distance \cite{robinson} between $T$ and $QMC(E(T))$. The Robinson-Foulds distance measures the number of bipartitions (also known as splits) of the taxa that appear in one tree but not the other. As a result RF distances tend to be larger for trees with more taxa. To account for this we report the Normalized Robinson-Foulds distance because it scales the RF distance by the maximum possible RF distance. Since QMC is a heuristic algorithm with no theoretical guarantees, it is both impossible to provide a proof that $QMC(E(T))=T$ for all trees $T$ and unreasonable to expect that the equality $QMC(E(T))=T$ would be observed in practice.

We used the $R$ package \emph{ape} to simulate unrooted, binary trees on increasing numbers of taxa and without assigned branch lengths \cite{ape}. For each tree size between $100$ and $1000$, we generated $1000$ trees under the Yule-Harding distribution with the \emph{rmtree} command in \emph{ape}. For each tree we computed the Normalized Robinson-Foulds distance between $T$ and $QMC(E(T))$. We also report the percentage of time when $QMC(E(T))$ is the original tree $T$.  The results of this study are displayed in Table~\ref{baseline}.

\begin{table}[tp]
\centering
\begin{tabular}{l|ccc}
\# Taxa & $T = QMC(E(T))$ & Normalized RF distance\\
\hline
100 & 99.5\% & 0.01\%\\
\hline
200 & 98.3\% &  0.02 \%\\
\hline
300 & 95.3\% &  0.07 \%\\
\hline
400 & 95.3\% &  0.05 \%\\
\hline
500 & 92.7\% &  0.06 \%\\
\hline
600 & 82.8\% &  0.10 \%\\
\hline
700 & 84.2\% &  0.10 \%\\
\hline
800 &85.0\% &  0.08 \%\\
\hline
900 & 78.8\% &  0.10 \%\\
\hline
1000 & 79.0\% &  0.09 \%\\
\hline
\end{tabular}
\caption{Comparison between $T$ and $QMC(E(T))$ for 1000 trees generated under the Yule-Harding model. Normalized RF distances are reported as the mean over all $1000$ trees in each data set.}
\label{baseline}
\end{table}

It is plausible that $QMC$ could return $T$ given even smaller input sets than $E(T)$. However, we ran this analysis using $QMC$ applied to linked systems $L(T)$, which are definitive but contain only $n-3$ quartets.  On the same $100$ taxa samples used in Table~\ref{baseline}, $QMC(L(T))$ returns the correct tree $0.0\%$ of the time and has a Normalized RF distance of $61.45\%$ in comparison to a $.01\%$ error observed for $QMC(E(T))$. 

Given that $QMC(L(T))$ fails to represent the combinatorial information in $T$ we do not report the similarly poor performance of $QMC(L(T))$ in our subsequent experiments. Our hypothesis is that the problem in applying $QMC$ to linked systems, or other small definitive quartet systems, is that not enough information remains in the divide-and-conquer step of QMC to reconstruct a tree. Therefore, the trees returned by $QMC(L(T))$ are largely unresolved. While it is possible to quickly reconstruct an accurate tree from $L(T)$ (p. 139 \cite{semple}) this does not appear to be easy using an algorithm which can also handle incompatible quartets.

\subsection{Application of an EQS in a Pipeline Using wQMC as a Summary Method} 
\label{sec:summary}
A \emph{summary method} is a method for estimating a species history on a set of $n$ species $X$ that has two steps.  First, gene trees are inferred from multiple loci using a tree-inference method.  Then a set of gene trees $\mathcal{G}$ from each locus are combined into a species tree $S$. See \cite{NJst} and  ASTRAL-II \cite{mirarab2015astral} for examples of promising (in terms of running time and accuracy) summary methods that will complete their analyses on datasets containing about 50 taxa in a \emph{reasonable} amount of time, e.g. less than 24 hours on a typical laptop or desktop machine purchased after 2011.

The quartet-agglomeration method wQMC \cite{WQMC} is a modification to QMC \cite{snir} that allows the quartets input to the method to be assigned weights by the user. Therefore wQMC can be used as a summary method, and this use of wQMC has been shown to yield reasonably accurate results on simulated datasets \cite{Davidson2015}. To use wQMC in this way, one first computes the set of all quartets $q$ that are displayed by each gene tree in a set of gene trees $\mathcal{G}$, and then computes the frequency with which $q$ appears in $\mathcal{G}$, which we denote by $w(q,\mathcal{G})$. 

A computational challenge arises in this approach due to the growth rate of the binomial coefficient ${n \choose 4}$. One approach to deal with this obstacle is to provide the quartet-agglomeration method with a selection of randomly sampled subsets of quartets. For example, this is the approach used in the experimental study in \cite{SVD} to complete an analysis on 52 taxa. This approach is also used in \cite{supertree, WQMC} and \cite{QMC}. But it is unknown what phylogenomic signal is lost when using randomly sampled quartets to represent the information in the dataset; this motivates the incorporation of an EQS representation of gene trees into summary method pipelines.

We demonstrate the use of EQS representations of gene trees by using wQMC as a summary method. We first compute EQS representations of each gene tree in a set $\mathcal{G}$, and then use the EQS representations of the trees in $\mathcal{G}$ to compute the quartet frequencies used as weights for wQMC. This approach was also used in \cite{Davidson2015}, but with using all quartets displayed by the gene trees. We chose a data set with 51 taxa from \cite{mirarab2015astral} for this experiment because

\begin{enumerate}
    \item datasets with this number of taxa are large enough to begin to pose computational difficulties such as those encountered in \cite{SVD} (in which they used QMC and ran most of their analyses on desktop machines without appealing to access to high-performance computing resources), and 
    \item datasets of interest to biologists regularly contain at least 50 taxa.  
\end{enumerate}

In particular, we use the dataset on 51 taxa from \cite{mirarab2015astral}, which was generated using the program SimPhy \cite{SimPhy}. This dataset is described in detail in the original paper \cite{mirarab2015astral} and linked to on the website for the supporting online materials for this paper. But briefly, we mention that the dataset contains 50 replicates, each containing 1000 true gene trees simulated on a model species tree under the MSC model as well as 1000 gene trees estimated from sequences simulated on the true gene trees. Our experiment uses the true gene trees to both \begin{enumerate}
\item maintain similarity with the baseline experiment in Section \ref{sec:baseline}, and
    \item avoid the introduction of gene tree estimation error in the first analysis of the EQS approach. Gene tree estimation error can only be bounded in the presence of a molecular clock \cite{esterror}, which cannot be assumed in the context of methods that use unrooted trees as inputs.  
    
\end{enumerate}

For the \emph{Control Version} of this experiment we first computed the set of all quartets displayed by each true gene tree in the set $\mathcal{G}$ for each replicate in the data set using custom scripts based on software developed for \cite{Johansen13computingtriplet}, and then computed the frequency $w(q,\mathcal{G})$ with which each quartet appeared in $\mathcal{G}$.  

In the \emph{Efficient Version} of this experiment, we tested the efficacy of an EQS for conserving the information of the gene trees in a summary method. We first found an EQS representing each true gene tree for each replicate in the data set, combined the resulting quartets for each replicate, which we denote as $Q_{E}$, and then computed the frequency $w_{E}(q,\mathcal{G})$ with which each quartet in this reduced set of quartets appeared. 

In each version of this experiment, the sets of frequencies from $w(q,\mathcal{G})$ (respectively $w_{E}(q,\mathcal{G})$) were given as user-assigned weights to quartets for wQMC and wQMC was used to infer the final species tree.

We measure accuracy for the summary method experiments using the proportion of splits in the simulated true species tree that are missing from the species trees estimated by the summary method or pipeline.  We refer to this as the \emph{missing branch rate}.

\subsubsection{Comparison of Experimental Results}

We note that the frequency of a quartet in $\mathcal{G}$ will clearly be lower using an EQS for each true gene tree in $\mathcal{G}$, but this only affects the weight of the quartet that must be sent to wQMC at the end of the species tree pipeline.  In addition, the upper bounds for the number of quartets which must be extracted for each gene tree is reduced from ${51 \choose 4}$ to ${49 \choose 2}$.

Using EQS representations of the trees in the set $\mathcal{G}$ did effectively reduce the number of quartets input to wQMC in our experiment. In the Control Version the average number of quartets across all 50 replicates given to wQMC as the input set was 649,474, while in the Efficient Version the average number of quartets across all 50 replicates given to wQMC as the input set was 155,511. This is a significant reduction in the number of quartets. As shown in Table \ref{tab:accuracy}, the large reduction in the number of quartets in the Efficient Version of the experiment did result in a minor reduction in accuracy. Our experiments indicate quartet-based summary method pipelines incorporating EQS representations of gene trees reduces the total number of quartets derived from the original gene trees but does not lead to a significant reduction in accuracy, therefore preserving most of the phylogenomic signal.

We re-ran the analyses of the dataset using ASTRAL-II version 4.7.8 to verify the results in the original paper \cite{mirarab2015astral}. Our mean missing branch rate across all 50 replicates for this dataset essentially matches that reported in \cite{mirarab2015astral}. However, in \cite{mirarab2015astral}, results were given in charts with labeled axes instead of precise numerical accuracies, which prevents exact comparison of the results.  We also compared the timing of our experiment to the timing of ASTRAL-II. ASTRAL-II does not require the pre-processing of quartets and follows a different algorithmic paradigm, so the only time reported is the species-tree estimation. Tables \ref{tab:real}, \ref{tab:user}, and \ref{tab:sys} show that the approximately six-fold reduction in the number of quartets sent to wQMC reduces the running time of wQMC by about 50\% in both real and system timing.  In addition, the accuracies shown in Table \ref{tab:accuracy} for the Efficient and Control Versions of our experiment are competitive with ASTRAL-II.  

This may have implications for other quartet-agglomeration methods. For example, another such method, Quartets FM (QFM), named after Fiduccia and Mattheyses (who introduced an algorithm known as FM for partitioning hypergraphs \cite{FM}), was introduced in \cite{QFM}. The simulations in \cite{QFM} showed improved accuracy over QMC but incurred a cost of a much slower running time. The reduction in time shown with EQS representations of trees in combination wQMC indicates that incorporating the use of an EQS may boost the timing performance of QFM in a similar pipeline. 

\begin{table}[tp]
\centering
\begin{tabular}{l|c | c}
Method & Real Timing & Real Timing \\ & (Quartets Pre-processing) & (Species-tree Estimation)  \\\hline
ASTRAL-II 4.7.8 & N/A & 36m50.100s \\\hline
Efficient Version &  689m29.315s & 0m39.604s \\\hline
Control Version & 875m2.449s  & 1m19.076s  \\\hline
\end{tabular}
\caption{Real Timing Results for EQS and wQMC vs. ASTRAL-II on a 51-taxon dataset from \cite{mirarab2015astral}. Timing data reported represents the mean time to complete all computations for a single replicate in the dataset.}
\label{tab:real}
\end{table}

\begin{table}[tp]
\centering
\begin{tabular}{l|c | c}
Method  & User Timing & User Timing \\ & (Quartets Pre-processing) &  (Species-tree estimation) \\\hline
ASTRAL-II 4.7.8  & N/A &  36m18.873s\\\hline
Efficient Version  &  608m51.966s & 0m36.769s \\\hline
Control Version & 2433m32.509s &  0m59.574s  \\\hline
\end{tabular}
\caption{User Timing Results for EQS and wQMC vs. ASTRAL-II 4.7.8 on a 51-taxon dataset from \cite{mirarab2015astral}.  Timing data reported represents the mean time to complete all computations for a single replicate in the dataset.}
\label{tab:user}
\end{table}

\begin{table}[tp]
\centering
\begin{tabular}{l|c | c}
Method  & System Timing & System Timing \\ & (Quartets Pre-processing) &  (Species-tree Estimation) \\\hline
ASTRAL-II 4.7.8 & N/A & 0m25.483s \\\hline
Efficient Version   & 41m53.718s & 0m0.315s   \\\hline
Control Version & 238m18.885s & 0m0.612s  \\\hline
\end{tabular}
\caption{System Timing Results for Efficient and Control Versions vs. ASTRAL-II 4.7.8 on a 51-taxon dataset from \cite{mirarab2015astral}.  Timing data reported represents the mean time to complete all computations for a single replicate in the dataset.}
\label{tab:sys}
\end{table}

\begin{table}[tp]
\centering
\begin{tabular}{l|c}
Method  & Accuracy \\\hline
ASTRAL-II 4.7.8 & 0.009576 \\\hline
Efficient Version & 0.015414  \\\hline
Control Version & 0.009992 \\\hline
\end{tabular}
\caption{Accuracy Results for Efficient and Control Versions vs. ASTRAL-II 4.7.8 on a 51-taxon dataset from \cite{mirarab2015astral}. Accuracy is given by the mean missing branch rate across all 50 replicates in the dataset.}
\label{tab:accuracy}
\end{table}

\subsection{Application of an EQS to the use of wQMC as a Supertree Method} 
\label{sec:supertree}

When reconstructing the evolutionary history of large and diverse samples of taxa, one must combine information from a variety of input trees into one large supertree reflecting the history of all taxa under consideration. Quartet-based algorithms such as wQMC can be used to combine these input trees into one large supertree. However, the MaxCut algorithm may fail to complete in a reasonable time if the number of taxa studied is over $500$ when using all quartets, or over $1000$ when using randomly sampled quartets \protect\cite{supertree}. Therefore, we restrict our analysis to the application of the MaxCut algorithm to an EQS-based pipeline rather than the full collection of quartets.  

An experimental methodology for testing supertree reconstruction algorithms was developed in \protect\cite{swenson2010simulation}. Swenson et al. used a sophisticated protocol to generate simulated input trees mimicking the process that a computational biologist would use to construct a supertree. For each supertree they constructed input trees which reflected the process of estimating $25$ clade-based trees and a single scaffold tree from DNA sequences which evolved along the corresponding induced subtree of the supertree. A scaffold tree contains a more disparate set of species and is meant to help glue the clade trees together. We emphasize that the clade and scaffold trees were not constructed under the MSC model, so if the gene trees are accurately estimated all quartets on these trees should have the same topology as the corresponding quartets on the supertrees. In this study (and in practice) there are errors in estimating the clade and scaffold trees from the sequence data.

We use the data from this study to test the accuracy of wQMC when applied to an EQS in the case when the true species tree has $1000$ taxa.  The density of the taxa which were included in the scaffold tree ranged from $20\%$ to $100\%$ \protect\cite{swenson2010simulation}.

Swenson et al. \cite{swenson2010simulation} compared MRP \protect\cite{ragan1992phylogenetic} using the incomplete species trees as inputs, and a combined analysis using maximum likelihood that reconstructed the species tree directly after concatenating the DNA sequence data. Methods were evaluated based on speed and on the Normalized RF distance rate between the reconstructed tree and the true species tree.

We did not investigate different weighting systems for wQMC in the supertree and baseline experiments, but instead if a quartet appears in an EQS representation of $l$ trees it receives a weight of $l$. Since the input data is in terms of unrooted tree topologies, in each case we analyze the results in terms of the topological distance between the model tree and the reconstructed tree.

Table~\protect\ref{tab:super} shows the Normalized RF distance rate when wQMC is applied to quartets derived from an EQS in comparison with the results found in \protect\cite{supertree}. Differences in computing power prevent a precise comparison between the running times published in \protect\cite{supertree} and wQMC applied used with an EQS-based approach. As a rough comparison, wQMC using an EQS to represent the input trees returns a supertree on $1000$ taxa in approximately $5$ minutes. The equivalent process was reported to take $1$ hour and $47$ minutes using MRP and almost $31$ hours when using the combined analysis with maximum likelihood \protect\cite{supertree}. We note that the algorithms have similar performance when the scaffold density is at least $50\%$, and the use of wQMC with an EQS-based approach results in decreased accuracy when the scaffold density is only $20\%$.

\begin{table}[tp]
\centering
\begin{tabular}{l|cccc}
Scaffold Density & $wQMC(E(T))$ & MRP* & Combined Analysis with ML* \\
\hline
$20\%$ & 43.2\% &23\% & 14\% \\
\hline
$50\%$ & 22.5\% & 21\% & 15\% \\
\hline
$75\%$ & 14.5\% & 18\% & 13\% \\
\hline
$100\%$ & 12.6\% & 15\% & 13\% \\
\hline
\end{tabular}
\caption{Normalized RF distance rate between true supertrees and trees reconstructed using wQMC with an EQS-based approach, MRP, and a combined analysis using maximum likelihood. Reported error is the mean over ten $1000$-taxon supertrees. Reported accuracies from both the MRP and Combined Analyses are estimated from Figure 5 of \protect\cite{swenson2010simulation}. }
\label{tab:super}
\end{table}

\section{Conclusion and Future Work}
\label{sec:conclusion}
As quartets remain a common input for summary and supertree methods one should carefully consider which sets of quartets best reflect the input data. In this article we demonstrate that an EQS theoretically encodes all of the data of the input tree, and in practice contains enough information for a fast heuristic algorithm to reconstruct over $99.9 \%$ of the topological data of the original input trees.

When an EQS is incorporated into a quartet-based summary method pipeline, the data loss from our initial study appears to be insignificant in terms of accuracy measured by the missing branch rate. Since computing the EQS representations of 50,000 gene trees (from the 50 replicates in the dataset) resulted in a large reduction of the number of quartets necessary to compute a species tree, this reduction has potential to assist other quartet-agglomeration techniques such as the QFM method and quartet-puzzling. As the size of datasets grows and newer methods such as SVDquartets 
\cite{SVD} require the combination of ever-increasing amounts of quartets, reduction in quartet set size has potential to enable the adoption of methods that can bypass gene tree estimation error.

When used in a supertree reconstruction, wQMC in combination with an EQS representation of the input trees has comparable performance to MRP and a combined analysis using Maximum Likelihood when the scaffold density is at least $50\%$. The performance is not as strong when the scaffold density is only $20\%$. The decrease in accuracy of wQMC when using a low scaffold density is offset by the dramatic increase in speed. This preliminary analysis shows that wQMC in combination with an EQS should be considered as a potential supertree reconstruction pipeline when large numbers of taxa need to be considered. However, care should be taken to ensure sufficient taxon coverage. It may also be possible that the low accuracy could be corrected by using multiple scaffold trees with lower taxon coverage.

As our timing data shows, pre-processing of quartets is a noteworthy component of quartet-based summary method pipelines such as the one presented in this paper. Analyzing all quartets is infeasible due to computational constraints in many situations, but using smaller definitive systems such as linked systems can be ineffective when using algorithms such as QMC and wQMC that must also handle incompatible quartets. The approach of incorporating an EQS lies in between and may be an appropriate standard tool when one needs to identify a representative set of quartets.  We conclude with open questions, which we believe may help to further improve the field.

\begin{ques}
Can one modify the QMC algorithm to ensure that $T=QMC(E(T))$ or that it has the property $T=QMC(f(T))$ where $f(T)$ is a definitive set of quartets on the order of $n$?
\end{ques}

\begin{ques}

Could quartet-based summary and supertree methods be improved by using weighting functions that account for potential error in gene tree estimation, or account for implementation-based biases in algorithms such as wQMC known to have no theoretical guarantees, but good performance in data simulation tests \cite{Davidson2015}?

\end{ques}

\begin{ques}\label{ques:QFM}
Recently, QFM has been re-implemented in the beta-testing version of PAUP* \cite{PAUP} with a refined implementation \cite{Swofford} over the original implementation in \cite{QFM}. Could QFM in combination with EQS produce better results than with wQMC? 
\end{ques}

\begin{ques}
Networks provide an alternative combinatorial framework to trees for describing complex evolutionary histories \protect\cite{bryant2004neighbor, huson2006application, huson2010phylogenetic}. Recent work suggests that displayed quartets can be an important tool in phylogenetic network reconstruction \protect\cite{mao2012quartet} which is more computationally demanding than the tree reconstruction considered here. Can the use of EQS representations of networks increase the scalability of quartet-based network reconstruction algorithms?
\end{ques}

\subsection{Description and Availability of Software}\label{Supporting}

Supporting materials, including the Efficient Quartets software developed by M.~Lawhorn and N.~Weber and the pipeline developed by R.~Davidson for incorporating the use of an EQS into a summary method pipeline, are available at the websites http://goo.gl/TSFzeD and https://github.com/redavids/efficientquartets.

\section*{Acknowledgment}
The authors wish to thank Laura Brunner for her assistance in running the supertree reconstruction analysis at the University of Wisconsin-Stevens Point. The authors also wish to thank the anonymous referees of this paper for many helpful comments and suggestions.

Research reported in this publication was supported by an Institutional Development Award (IDeA) from the National Center for Research Resources (5 P20 RR016461) and the National Institute of General Medical Sciences (8 P20 GM103499) from the National Institutes of Health. R.D. was supported by NSF grant DMS-1401591. This material is based upon work supported by the National Science Foundation under Grant No. DMS-1616186.

\ifCLASSOPTIONcaptionsoff
 \newpage
\fi

\bibliographystyle{IEEEtran}
\bibliography{sample}

\end{document}